\documentclass[sigconf]{acmart}
\renewcommand\footnotetextcopyrightpermission[1]{} 
\makeatletter
\newsavebox{\@brx}

\newcommand{\llangle}[1][]{\savebox{\@brx}{\(\m@th{#1\langle}\)}%
  \mathopen{\copy\@brx\kern-0.5\wd\@brx\usebox{\@brx}}}
\newcommand{\rrangle}[1][]{\savebox{\@brx}{\(\m@th{#1\rangle}\)}%
  \mathclose{\copy\@brx\kern-0.5\wd\@brx\usebox{\@brx}}}
\makeatother
\usepackage[utf8]{inputenc} 
\usepackage[T1]{fontenc} 
\usepackage{booktabs} 
\usepackage{multirow}
\usepackage{booktabs}
\usepackage[normalem]{ulem}
\usepackage{tikz}
\usepackage{enumitem}
\usepackage{adjustbox}
\usepackage{relsize}
\usepackage{amsthm}
\usepackage{amsmath}
\usepackage{float}
\usepackage{textcmds}
\usepackage{tabularx}
\usepackage{listings}
\usepackage{url}
\usepackage{enumitem}
\usepackage{multirow}
\usepackage{eqnarray}
\usepackage{chngcntr}
\usepackage{hyperref}
\usepackage{amsmath}

\setcopyright{none}
\acmConference[SAC'20]{ACM SAC Conference}{March 30-April 3, 2020}{Brno, Czech Republic} 
\acmYear{2020}

\begin{document}
\title{Formal Verification of Debates in Argumentation Theory}


\author{Ria Jha}
\affiliation{%
  \institution{Imperial College London, UK}
}
\email{ria.jha15@imperial.ac.uk}

\author{Francesco Belardinelli}
\affiliation{%
  \institution{Imperial College London, UK \& Universit\'e d'Evry, France}
}
\email{fbelard@ic.ac.uk}
\author{Francesca Toni}
\affiliation{%
  \institution{Imperial College London, UK}
}
\email{f.toni@imperial.ac.uk}

\begin{abstract}
Humans engage in informal debates on a daily basis. By expressing
their opinions and ideas in an argumentative fashion, they are able to
gain a deeper understanding of a given problem and in some cases, find the
best possible course of actions towards resolving it.
In this paper, we develop a methodology to
verify debates formalised as abstract argumentation frameworks. We first present a translation from debates to
transition systems. Such transition systems can model debates and
represent their evolution over time using a finite set of states.
We then formalise relevant debate properties using temporal and
strategy logics. These formalisations, along with a debate transition
system, allow us to verify whether a given debate satisfies
certain properties. The verification process can be automated using
model checkers. Therefore, we also measure their performance
when verifying debates, and use the results to discuss the
feasibility of model checking debates.
\end{abstract}
%
%


\maketitle


  

\section{Introduction}
Humans engage in informal debates on a daily basis. By expressing
their opinions and ideas in an argumentative fashion, they are able to
gain a deeper understanding of a given problem and, in some cases,
find the best possible course of actions towards resolving it.

Verifying a debate gives us the ability to learn about its overall
outcome. It allows us to reason about the strategies available to a
participant as well as to determine the acceptability of some argument
they have made. Given a debate, we can verify relevant properties such
as ``the proponent will be able to refute any attack from the
opponent'' or ``the proponent has a strategy such that it will
eventually win the debate''.

A key technique for verifying properties of debates is model
checking. Model checking is a verification technique that has been developed over the last thirty years \cite{baier2008principles,ClarkeGrumbergPeled99}. This process includes developing and examining the full model of a system to ensure that it
satisfies a specific property, which is normally expressed in some
logic-based
language.  Typically, this procedure is completely automated.
While there are model checking tools and techniques that have been
developed for verification of properties of general transition systems
\cite{Cimatti+02b,Lomuscio2017}, to the best of
our knowledge, methodologies for the formal
verification of debates have not yet been considered .

In this paper, we develop a novel methodology to verify debates
formalised as abstract argumentation frameworks. We first present a
translation from debates to transition systems. Such transition
systems can model debates and represent their evolution over time
using a finite set of states.
We then formalise relevant debate properties using various flavours of
temporal logics. These formalisations, along with a debate transition
system, allow us to verify whether a given debate is able to satisfy
certain properties. The verification process can be automated using
model checkers. Thus, we also measure the performance of model
checkers when verifying debates, and use the results to discuss the
feasibility of model checking debates.

\textbf{Related work.}  On the formal analysis of debates in
Argumentation Theory, \cite{modgil2009proof} and
\cite{dung2007computing} consider debates between two agents, the
proponent and the opponent. In \cite{modgil2009proof} a debate between
the two agents is formalised as a turn-based game with the purpose of
determining the acceptability of the initial argument put forward by
the proponent. An agent may put forward all the arguments that can be
made legally according to the rules of the game, in order to refute
its counterpart’s argument.  The proponent winning such a game, where
all arguments have been made legally, indicates that its initial
argument is acceptable. The debates introduced in
\cite{dung2007computing} can be used for the same purpose. Unlike
\cite{modgil2009proof}, the proponent must behave deterministically
and only select a single argument to refute its counterpart. Moreover,
\cite{dung2007computing} introduces several debate properties which, if
satisfied, indicate the initial argument's acceptability.

In this paper, we present a translation for the debates in
\cite{modgil2009proof} into transition systems. We then formalise
certain debate properties \cite{dung2007computing} so that they can be
interpreted on these transition systems. These formalisations take
into consideration the notion of a deterministic proponent, as
introduced in \cite{dung2007computing}. This enables us to reason
about the different strategies available to the proponent in a debate.

There has not been much work done in the development of formal
methods for the verification of debates. \cite{belardinelli2015formal} models debates between
agents, where each agent holds private, possibly infinite, argumtentation frameworks. Thus, agents are able to exchange arguments and build a public
framework. Then, \cite{belardinelli2015formal} looks at formally
expressing and verifying relevant properties of abstract argumentation
frameworks.
However, it does not consider the notion of dispute trees nor acceptability conditions, as we do in
this paper.

Other works include \cite{GrossiD.2010Otlo}, where abstract
argumentation frameworks are treated as Kripke frames \cite{Fagin+95b}. \cite{GrossiD.2010Otlo} then uses modal logic to formalise notions
of argumentation theory including conflict-freeness and admissibility.
However, as with \cite{belardinelli2015formal}, no attempts are made to formally verify the acceptability of arguments.

\textbf{Structure of the Paper.} In Sec.~\ref{absarg} we present the basics of
Argumentation Theory \cite{dung1995acceptability}, including the
various semantics, as well as the notion of dispute tree introduced in
\cite{modgil2009proof}; while Sec.~\ref{modcheck} is devoted to the
preliminaries on model checking temporal and strategy logics.  In
Sec.~\ref{Sec3} we define a translation from debates in Argumentation
Theory \cite{modgil2009proof} to Interpreted Systems \cite{Fagin+95b}, and in Sec.~\ref{Sec4}
we formalise various winning conditions as formulas in Strategy Logic
\cite{MogaveroF.2014RasO}. Finally, in Sec.~\ref{Sec5} we evaluate the
performance of our approach against state-of-the-art argumentation
reasoners. We conclude in Sec.~\ref{conc} and discuss directions for
future work.

%

\section{Preliminaries}

In this section we present preliminary materials on abstract
argumentation (Sec.~\ref{absarg}) and verification by model checking
(Sec.~\ref{modcheck}), which will be used in the rest of the paper.

\subsection{Abstract Argumentation} \label{absarg}

To provide debates with a formal vest, we consider \textit{abstract
argumentation frameworks} as introduced
in \cite{dung1995acceptability}.
\begin{definition}[\textbf{AF}] An \textit{(abstract) argumentation framework} is a pair $\mathcal{AF}= \langle Args, Att \rangle$ where: 
\begin{itemize}
    \item \textit{Args} is a set of arguments $a, b, c, \ldots$.

The
    internal structure of each argument $a \in Args$ is abstracted and
    the arguments are perceived as atomic
    entities.

\item $Att \subseteq Args \times Args$ is a binary (attack) relation
    on $Args$.
\end{itemize}
\end{definition}

Given an argumentation framework, it is possible to compute the
``acceptable'' sets of arguments, referred to
as \textit{extensions}. Notably, \cite{dung1995acceptability,
dung2007computing} present several extension-based semantics. An
extension $E \subseteq Args$ computed under a particular semantics
must fulfil some specific criteria. We start with the notion of
acceptable argument.

\begin{definition}[\textbf{Acceptability}]
An argument $a \in Args$ is \textit{acceptable} with respect to set
$E \subseteq Args$ iff for each argument $ b \in Args$ that attacks
$a$, there is some argument $c \in E$ that attacks $b$. In other
words, \textit{a} is defended by \textit{E}.
\end{definition}
In Table~\ref{semantics} we
present the different semantics and the criteria that a given set $
E \subseteq Args$ must fulfill under each of
them \cite{dung1995acceptability, dung2007computing}.
\begin{table}
  \begin{tabular}{|lcl|}
   \hline
 \textbf{Semantics} & & \textbf{Criteria} \\
 \hline
 conflict-free  & iff &
 there are no
$a, b \in E$ such that $(a,b) \in Att$\\
 \hline
 admissible & iff &  it is conflict-free\\
 & and &  every argument in $E$ is acceptable w.r.t.~$E$.\\
 \hline
 complete & iff & it is admissible \\
          & and & it contains all arguments acceptable w.r.t.~$E$\\
 \hline
          grounded    & iff &  it is complete\\
          & and & minimal w.r.t.~set inclusion \\ 
 \hline
          preferred   & iff &  it is admissible\\
          & and & maximal w.r.t.~set inclusion \\ 
 \hline
          ideal      & iff & it is admissible\\
          & and & a subset of every preferred extension\\
 \hline
  \end{tabular}
  \caption{Semantics of argumentation frameworks. \label{semantics}}
\end{table}

\begin{example}
\label{ex: aaf}
As an example of an argumentation framework, consider
$\mathcal{AF}$ =
$\langle$\{\textit{a,b,c,d,e,f,g}\}, \textit{Att}$\rangle$\
with \textit{Att} = \{(\textit{a, b}), (\textit{b, a}),
(\textit{b, b}), (\textit{d, c}), (\textit{e, f}), (\textit{f, e}),
(\textit{f, d}), (\textit{g, f})\}.

An argumentation framework can also be depicted as a graph with nodes
and directed edges. The nodes represent the arguments and the directed
edges represent an attack from one argument to another. See
Fig.~\ref{fig:aaf} for a representation of framework
$\mathcal{AF}$. Given this framework, we can compute the acceptable sets
of arguments under the different semantics:
\begin{itemize}
 \item The complete extensions are \{\textit{c, e, g}\}, \{\textit{a, c, e, g}\}.
 \item The grounded extension is \{\textit{c, e, g}\}.
    \item The preferred extension is \{\textit{a, c, e, g}\}.
    \item The ideal extension is \{\textit{a, c, e, g}\}.
\end{itemize}
\end{example}

\begin{figure}[H]
    \centering
    \begin{tikzpicture}[state/.style={circle, draw, minimum size=0.7cm}, scale=0.8]
    \node[state] (a) at (0,0) {$a$};
    \node[state] (b) at (0,-2) {$b$};
    \node[state] (c) at (2,0) {$c$};
    \node[state] (d) at (4,0) {$d$};
    \node[state] (e) at (6,0) {$e$};
    \node[state] (f) at (6,-2) {$f$};
    \node[state] (g) at (8, -2) {$g$};
    \path [->] (a.280) edge (b.80);
    \path [->] (b.100) edge (a.260);
    \path [->] (b) edge [loop right] (b);
    \path [->] (d) edge (c);
    \path [->] (e) edge (d);
    \path [->] (f) edge (d);
    \path [->] (e.280) edge (f.80);
    \path [->] (f.100) edge (e.260);
    \path [->] (g) edge (f);
\end{tikzpicture}
    \caption{The abstract argumentation framework $\mathcal{AF}$ in Example~\ref{ex: aaf}.}
    \label{fig:aaf}
\end{figure}
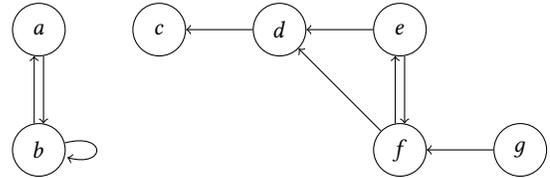

\label{sec:disp_tree}
The acceptability of some argument $a$ under the grounded, preferred,
or ideal semantics can be determined by constructing a \textit{dispute
tree}, where $a$ is the root argument \cite{modgil2009proof}. Each
branch in a dispute tree is a sequence of attacking arguments, which are put forward by the proponent and the opponent in a turn based fashion. A branch in a dispute tree is finite when the last played argument is unattacked. Otherwise, the branch may be infinitely long.

\begin{definition}[\textbf{Dispute Tree}]
\label{def:disp_gen}
Given an argumentation framework $\mathcal{AF} =
\langle Args, Att\rangle$ and an argument $a \in
Args$, a \textit{dispute tree} $T$ induced by $a$ is a tree
of arguments where
\begin{enumerate}
    \item The root node of \textit{T} is \textit{a}, played by the proponent.
    \item For all $x, y \in Args$, \textit{x} is a child of \textit{y} iff $(x, y) \in Att$.
    \item Each node is labelled either \textbf{P} or \textbf{O}, for the proponent or the opponent, indicating the player that put forward the argument. 
\end{enumerate}
\end{definition}

\begin{example}
Consider the argumentation framework $\mathcal{AF}$ depicted in Fig.~\ref{fig:aaf}. The argument $c$ induces a single dispute tree using Def.~\ref{def:disp_gen}, which is shown in Fig.~\ref{fig:first_disp_tree}.  
\end{example}
\begin{figure}
    \centering
     \begin{tikzpicture}[level distance=1.05cm,
  level 1/.style={sibling distance=3cm},
  level 2/.style={sibling distance=1.8cm},
  level 3/.style={sibling distance=1cm},grow'=right]
  \node {\textbf{P:}  \textit{c}}
    child {node {\textbf{O:}  \textit{d}}
      child {node {\textbf{P:}  \textit{e}}
            child{node {\textbf{O:}  \textit{f}}
                child{
                node{\textbf{P:}  \textit{e}}
                    child{node{}edge from parent[dashed]}}
                child{node{\textbf{P:}  \textit{g}}}}}
      child {node {\textbf{P:}  \textit{f}}
      child{node{\textbf{O:}  \textit{g}}}
      child{
                node{\textbf{O:}  \textit{e}}
                    child{node{\textbf{P:}  \textit{f}}
                    child{node{}edge from parent[dashed]}
                    child{node{}edge from parent[dashed]}}
                    }
                    }
    };
\end{tikzpicture}
    \caption{The dispute tree induced by argument $c$, generated using
    $\mathcal{AF}$ in Fig.~\ref{fig:aaf}.}  \label{fig:first_disp_tree}
\end{figure}
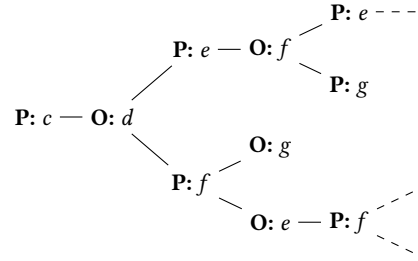
Within such a dispute tree, we can find subtrees that are the result
of an agent applying one of its \textit{strategies}. 
\begin{definition}[\textbf{Strategy}]
Given $\mathcal{AF} =
\langle Args, Att\rangle$, a \textit{strategy}
$\sigma: Args \rightharpoonup Args$ is a partial function that maps an
argument $x \in Args$ to one of its attackers, if there are
any.
\end{definition}

In a dispute tree $T$, induced by argument $a$, the
proponent can apply some strategy $\sigma$ from $a$ to obtain a subtree
$T_\sigma$ where:
\begin{itemize}
    \item The root argument of $T_\sigma$ is $a$.
    \item For any argument $i$ played by the opponent, the proponent must play at most one argument, $\sigma(i)$.
\end{itemize}

To justify \textit{a}'s membership to some extension, the proponent
must have a \textit{winning strategy} $\sigma$. By applying $\sigma$
from $a$, the proponent can guarantee acceptability of $a$ under a
particular semantics.

Definitions \ref{def:grd_win}, \ref{def:adm_win},
and \ref{def:ideal_win} define the grounded, admissible, and ideal
winning strategies, respectively \cite{dung2007computing}. For these
definitions, consider a dispute tree $T$, induced by some argument
$a$.

\begin{definition}[\textbf{Grounded Winning Strategy}]
\label{def:grd_win}
A strategy $\sigma$ applied at root $a$ is a {\em grounded winning strategy} iff:
\begin{itemize}
        \item Every opponent node in $T_\sigma$ has exactly one child.
        \item Every dispute within $T_\sigma$ is finite. 
\end{itemize}
\end{definition}
To justify membership of $a$ in a preferred extension, it suffices to show that $a$ belongs to some admissible extension \cite{modgil2009proof}.
\begin{definition}[\textbf{Admissible Winning Strategy}]
\label{def:adm_win}
A strategy $\sigma$ applied at root $a$ is an {\em admissible winning strategy} iff:
\begin{itemize}
        \item Every opponent node in $T_\sigma$ has exactly one child.
        \item No argument is labelled by both the proponent and the opponent. 
\end{itemize}
\end{definition}
\begin{definition}[\textbf{Ideal Winning Strategy}]
\label{def:ideal_win}

A strategy $\sigma$ applied at root $a$ is an {\em  ideal winning strategy} iff:
\begin{itemize}
    \item $\sigma$ is an admissible winning strategy.
    \item There does not exist an admissible winning strategy $\sigma'$ for any argument $b \in Opp(T_\sigma)$, which is the set of all arguments played by the opponent in $T_\sigma$.
\end{itemize}
\end{definition}
\begin{example}
\label{ex:disp_two}

Consider the argumentation framework in Fig.~\ref{fig:aaf} and the
dispute tree $T$ in
Fig.~\ref{fig:first_disp_tree}, which
shows that the proponent has a choice when attacking arguments $d$ and
$f$. As a result, there are three strategies available to the
proponent: $\sigma_1, \sigma_2$, and $\sigma_3$. By applying each of
these strategies from $c$, we obtain the subtrees $T_{\sigma_1},
T_{\sigma_2}$, and $T_{\sigma_3}$, respectively. The dispute tree
$T_{\sigma_1}$ is depicted in Fig.~\ref{fig:disp_two_tree_one} whereas
$T_{\sigma_2}$ and $T_{\sigma_3}$ are depicted in
Fig.~\ref{fig:disp_two_tree_two} and \ref{fig:disp_two_tree_three},
respectively. We find that the dispute trees $T_{\sigma_1}$ and
$T_{\sigma_3}$ satisfy the properties described in
Def.~\ref{def:adm_win} and therefore $\sigma_1$, and $\sigma_3$
are both admissible winning strategies. Whereas, only $T_{\sigma_3}$
satisfies the properties in Def.~\ref{def:grd_win} and
so $\sigma_3$ is the only grounded winning strategy. Finally,
according to Def.~\ref{def:ideal_win}, $\sigma_1$ and $\sigma_3$
are also ideal winning strategies. 
\end{example}
\begin{figure}
    \centering
     \begin{tikzpicture}[level distance=1.05cm,
  level 1/.style={sibling distance=3cm},
  level 2/.style={sibling distance=1.8cm},
  level 3/.style={sibling distance=1cm}, grow'=right]
  \node {\textbf{P:}  \textit{c}}
    child {node {\textbf{O:}  \textit{d}}
      child {node {\textbf{P:}  \textit{e}}
            child{node {\textbf{O:}  \textit{f}}
                child{node{\textbf{P:}  \textit{e}}
                child{node{}edge from parent[dashed]}}}}
    };
\end{tikzpicture}
    \caption{The sub dispute tree $T_{\sigma_1}$ induced by $c$ in Ex.~\ref{ex:disp_two}.}
    \label{fig:disp_two_tree_one}
    \centering
    \begin{tikzpicture}[level distance=1.05cm,
  level 1/.style={sibling distance=3cm},
  level 2/.style={sibling distance=1.8cm},
  level 3/.style={sibling distance=1cm}, grow'=right]
  \node {\textbf{P:}  \textit{c}}
    child {node {\textbf{O:}  \textit{d}}
      child {node {\textbf{P:}  \textit{f}}
      child{node{\textbf{O:}  \textit{g}}}
      child{
                node{\textbf{O:}  \textit{e}}
                    child{node{\textbf{P:}  \textit{f}}
                    child{node{}edge from parent[dashed]}
                    child{node{}edge from parent[dashed]}}
                    }
                    }
    };
    \end{tikzpicture}
    \caption{The sub dispute tree $T_{\sigma_2}$ induced by $c$ in Ex.~\ref{ex:disp_two}.}
    \label{fig:disp_two_tree_two}
    \centering
    \begin{tikzpicture}[level distance=1.05cm,
  level 1/.style={sibling distance=3cm},
  level 2/.style={sibling distance=1.8cm},
  level 3/.style={sibling distance=1cm}, grow'=right]
  \node {\textbf{P:}  \textit{c}}
    child {node {\textbf{O:}  \textit{d}}
      child {node {\textbf{P:}  \textit{e}}
            child{node {\textbf{O:}  \textit{f}}
                child{node{\textbf{P:}  \textit{g}}
              }}}
    };
\end{tikzpicture}
    \caption{The sub dispute tree $T_{\sigma_3}$ induced by $c$ in Ex.~\ref{ex:disp_two}.}
    \label{fig:disp_two_tree_three}
\end{figure}


\subsection{Verification by Model Checking} \label{modcheck}


An interpreted system is a transition system that allows us to reason
about the behaviour
and strategies of agents in a multi agent system  \cite{lomuscio2006model}.
Given a set \textit{$\Sigma$ =} \{1,.., \textit{n}\} of agents  and a special agent \textit{Env} for the environment, let $Ag$ = $\Sigma \cup \{Env\}$.
\begin{definition}[\textbf{Interpreted System}]
\label{def:int_sys}
 An \textit{interpreted system} is a tuple $IS = \langle (L_i, Act_i,
 P_i, t_i)_{i \in Ag}, I, h\rangle$ where:
\begin{itemize}
    \item For each agent $i \in \Sigma$, $L_i$ is the set of all possible local states of $i$. The local state of each agent $i \in \Sigma$ is private. For agent \textit{Env}, its set of local states is referred to as $L_{Env}$ and may be observed by the other agents. We refer to the local states of the agents in $Ag$ collectively as a global state $g \in L_1 \times \ldots \times L_n \times L_{Env}$.
    
    \item For each agent $i \in Ag$, $Act_i$ is the set of all actions available to $i$. The set $Act =  Act_1 \times \ldots \times Act_n \times Act_{Env}$ refers to the joint actions for all the agents in $Ag$.
    
    \item Protocol \textit{P$_i$} : \textit{L$_i \times L_{Env}$ $\rightarrow$} $2^{Act_i}$ is a function for an agent $i \in \Sigma$, which takes the local states of \textit{i} and \textit{Env} as its input and return all actions for \textit{i} that are enabled at the given state.
    
    \item $t_i : L_i \times L_{Env} \times Act \rightarrow L_i$ is a
    deterministic transition function for agent $i \in \Sigma$, which
    describes the evolution of \textit{i}’s local state. Using the local states of \textit{i} and \textit{Env}, along with the
    enabled actions of all the agents, it returns the ``next'' local
    state for agent \textit{i}. In particular, $t_i(l_i, l_{Env},\alpha)$ is defined iff $\alpha_i \in P_i(l_i, l_{Env})$.
Similarly, the transition function for the environment \textit{Env} is
    $t_{Env} : L_{Env} \times Act \rightarrow L_{Env}$. A global
    transition function $t: G \times Act \rightarrow G$, where $G \subseteq L_1 \times \ldots \times L_n \times L_{Env}$ refers to the set of \textit{reachable} global states, combines the
    output of the local transition functions and provides the next
    accessible global state of the system. In particular, $t(g,\alpha)
    = \prod_{i \in Ag} t_i(g_i, g_{Env}, \alpha)$.
    
    \item $I \subseteq L_1 \times \ldots \times L_n \times L_{Env}$
     is the initial set of global states. The
    set $G \subseteq L_1 \times \ldots \times L_n \times L_{Env}$ refers to all the global states that
    are \textit{reachable} from $I$ through the transition
    function $t$.
    
    \item$h: AP \rightarrow 2^G$, where $AP$ is a set of atomic propositions, is a valuation function.
\end{itemize}
\end{definition}

Strategy Logic (SL) is a logic used to reason about strategies of agents in multi agent systems \cite{MogaveroF.2014RasO}.
Consider an interpreted system $IS$, its agents $Ag$, and fix a set
$Var$ of variables for strategies. Each variable $x_i \in Var$ is
typed according to an agent $i \in Ag$.
\begin{definition}[\textbf{SL syntax}]
\label{def: sl_syntax} The syntax of SL is presented below:
\begin{center}
$\phi \Coloneqq$ \textit{ p } $\mid$ $\neg \phi \mid  \phi \lor \phi \mid X \phi \mid G \phi \mid \phi U \phi \mid $ $\exists x_i \phi $
\end{center}
\end{definition}
The strategy quantifier $\exists x_i$ is read as ``there
exists a strategy $x_{i}$'' while $\forall x_{i}
::= \neg \exists x_i \neg$ can be read as ``for all
strategies $x_{i}$''.  The intuitive meaning of the linear-time
operators is standard \cite{baier2008principles}: $X \phi$ is read as
``at the next moment $\phi$ holds'', $G \phi$ as ``always $\phi$'',
and $\phi U \phi'$ as ``$\phi$ until $\phi'$. The operator $F \phi$ is read as ``$\phi$ eventually holds'', which may be equivalently expressed as $\top U \phi$.
In interpreted systems, each agent $i \in \Sigma$ is assumed to have
its own strategy $f_i: L_i \times L_{Env} \rightarrow Act_i$ such that
$f_i(l_i,l_e) \in P_i(l_i,l_e)$. Similarly, \textit{Env} can also have
a strategy $f_{Env}: L_{Env} \rightarrow Act_{Env}$ to determine its
next action. $Str_i$ refers to the set of strategies available to
agent $i \in Ag$. We can use an {\em assignment} to associate a
variable $x_{i} \in Var$, for $i \in Ag$, to one of its strategies
$f_i$. An assignment $\chi: Var \to Str$, for $Str = \bigcup_{i \in
Ag} Str_i$, is a function mapping
each variable $x_{i} \in Var$ to a strategy $f_i \in Str_i$.
Given $\chi$, $x_{i} \in Var$,
and $i$'s strategy
$f_i$, $\chi [x_{i} \mapsto f_i]$
denotes the new
assignments on $dom(\chi) \cup \{x_{i}\}$.

From an assignment $\chi$ and a global state $s$ in
$IS$, we can obtain an outgoing path $\lambda$ from $s$, known as a
$play$. A $play$ is the unique possible outcome obtained from each agent
$i \in Ag$, applying the strategy that they are assigned to
in $\chi$, from $s$. We
use $\chi$-$s$ to denote a play from a global state
$s$, with respect to $\chi$.
\begin{definition}[\textbf{SL semantics}] Given $IS$, its current global state $s$, an assignment \raisebox{1.25pt}{$\chi$}, an SL formula $\phi$,
the semantics of SL are given below. Note that for a play $\lambda$,
$\lambda[i]$ refers to the global state occurring in the $i$-th
position in $\lambda$, for $i \geq 0$.
\begin{tabbing}
$(IS, \raisebox{1.25pt}{$\chi$}, s) \models p$ \ \ \ \ \ \ \ \ \ \ \ \ \ \= iff \ \ \= $s \in h(p)$\\
$(IS, \raisebox{1.25pt}{$\chi$}, s) \models \neg \phi$ \> iff \> $(IS, \raisebox{1.25pt}{$\chi$}, s) \not \models \phi$.\\
$(IS, \raisebox{1.25pt}{$\chi$}, s) \models \phi_1 \lor \phi_2$ \> iff \>  $(IS, \raisebox{1.25pt}{$\chi$}, s) \models \phi_1$ or $(IS, \raisebox{1.25pt}{$\chi$}, s) \models \phi_2$.\\
$(IS, \raisebox{1.25pt}{$\chi$}, s) \models \exists
x_{i} \phi$ \> iff \> for some strategy $f_i \in Str_i$,\\
\> \> $(IS, \chi[x_{i} \mapsto f_i], s) \models \phi$\\
$(IS, \raisebox{1.25pt}{$\chi$}, s) \models X\phi$ \> iff \> $(IS, \raisebox{1.25pt}{$\chi$}, \lambda[1]) \models \phi$\\
$(IS, \raisebox{1.25pt}{$\chi$}, s) \models G\phi$ \> iff \> for all $i \geq 0$, $(IS, \raisebox{1.25pt}{$\chi$}, \lambda[i]) \models \phi$\\
$(IS, \raisebox{1.25pt}{$\chi$}, s) \models \phi_1 U \phi_2$ \> iff \>
for some $i \geq 0$,
$(IS, \raisebox{1.25pt}{$\chi$}, \lambda[i]) \models
\phi_2$ and\\
\> \> for all $0 \leq j < i$, $(IS, \raisebox{1.25pt}{$\chi$}, \lambda[j]) \models \phi_1$
\end{tabbing}
\end{definition}
A sentence $\phi$ is satisfied in $IS$ if
for all initial global states $g \in I$, $(IS, \emptyset,
i) \models \phi$, where $\emptyset$ refers to an empty assignment.

\label{logic: ATL}
Alternating-Time Temporal Logic (ATL) is a branching time temporal
logic that is used for strategic reasoning in a multi agent
system \cite{Alur:2002:ATL:585265.585270}.
ATL can be seen as a syntactic fragment of SL, where formulas
$\llangle A \rrangle \phi$, for $A = \{i_1, \ldots, i_k\} \subseteq
Ag$, is defined as $\exists x_1 \ldots \exists x_{k} \forall
x_{k+1} \allowbreak \ldots \forall x_{n} \phi$:
``the coalition $A$ of agents has a joint strategy to enforce some
property $\phi$, regardless of the behaviour of the other agents in
$Ag \setminus A$''.

Computation Tree Logic (CTL) is another branching-time temporal logic
that is strictly less expressive than
ATL \cite{baier2008principles}. CTL can be considered as a syntactic
fragment of ATL as it only allows quantification either over all or
over a single path in an interpreted system. Within CTL, the
quantifier $A$ refers to all paths in the system, and is equivalent to
path quantifier $\llangle \emptyset \rrangle$ in ATL. The
quantifier $E$ refers to the existence of an individual path in $IS$
and is equivalent to the ATL quantifier $\llangle Ag \rrangle$, where
$Ag$ is the set of all agents within the interpreted system.

Given a finite-state model \textit{IS} of a system and a formal
property $\phi$ that the system must satisfy, \textit{model checking}
is an automated verification technique that can be used to check
whether $IS \models \phi$ \cite{baier2008principles}, where
$IS \models \phi$ iff for all $g \in I$, $(IS, g) \models \phi$.
There are several model checking tools available that can be used for
the verification of multi-agent systems, including
MCMAS \cite{Lomuscio2017, mcmas_manual}, which is a model checker
tailored on interpreted systems.  A variant of MCMAS,
MCMAS-SLK \cite{CermakP2014MAmc}, will be used in the experimental
evaluation in Sec.~\ref{Sec5}.


\section{Constructing Interpretation Systems from Debates}
\label{Sec3}
In order to verify a debate using model checking, a transition system
representing the debate is required. The transition system should be
able to represent the debate's state at all points in time and also
represent its evolution. Once the transition system has been
generated, it can then be used to verify certain properties of the
debate.

\label{sec:fst_appch}

This section provides a translation for a dispute tree $T$, as
described in Def.~ \ref{def:disp_gen}, to an interpreted system $IS_T$
as per Def.~\ref{def:int_sys}. An advantage of this approach is that
we can use a single interpreted system $IS_T$ to verify the existence
of a winning strategy $\sigma$ in $T$ for the proponent, under
multiple semantics. We can do this in one step by formally expressing
the properties of the subtree $T_\sigma$ under the different
semantics, and checking whether $IS_T$ satisfies them.

In order to construct $IS_T$, we require an abstract argumentation
framework $\mathcal{AF}= \langle Args, Att \rangle$, where $Args$ is
finite. Furthermore, we require the dispute tree $T$, induced by some
argument $a \in Args$, whose acceptability we would like to determine.
\begin{definition}[\textbf{The Interpreted System $IS_T$}] \label{IST}
Given an abstract argumentation framework $\mathcal{AF}= \langle Args,
Att \rangle$ and corresponding dispute tree $T$ starting in argument
$a \in Args$, the interpreted system $IS_T = \langle (L_i, Act_i, P_i,
t_i)_{i \in Ag_T}, I_T, h_T \rangle$ is comprised of:
\begin{itemize}
  \item Set $Ag_T = \{ Pro, Opp, Env \}$ of agent including the proponent \textit{Pro}, the opponent \textit{Opp}, and the environment, \textit{Env}.

\item For every agent in $Ag_T$, the set of local states are:
\begin{itemize}
    \item $L_{Pro}$ = \textit{Args};
    \item $L_{Opp} = Args \cup  \left\{empty\right\}$;
    \item $L_{Env} = \{Pro, Opp\} \times Args \times AttSeen$, where $AttSeen \subseteq 2^{Att}$
\end{itemize}

The unique initial state of the system is $(a, empty, (Opp, a,
\emptyset))$, indicating that argument $a$ has been put forward by
\textit{Pro}. Therefore, $I_T = \{(a, empty, (Opp, a, \emptyset))\}$.

\item The set of actions available to agents \textit{Pro} and \textit{Opp} is $Act_{Pro}$ = $Act_{Opp}$ =  $\{attack_{xy} \mid (x,y) \in Att\} \cup \{nothing\}$. Whereas, the set of actions available to \textit{Env} is $Act_{Env}$ = $\{nothing\}$.

\item The protocol functions for \textit{Pro} and \textit{Opp} determine the
  available actions that the agents may perform, based on their current local state and the local state of environment \textit{Env}:
\[
    P_{Pro}(\_, \textit{(turn, last, \_)})= 
\begin{cases}
    attacks(last), & \text{if $turn = Pro$ and} \\ & \text{\textit{attacks}(\textit{last})$\ne \emptyset$} \\
    \{nothing\} & \text{otherwise}
\end{cases}
\] where \textit{attacks(y)} = $\{attack_{xy} \mid attack_{xy} \in Act_{Pro}\}$ and $\_$
indicates that the input provided is not used in the function.

Symmetrically,
\[
    P_{Opp}(\_, \textit{(turn, last, \_)})= 
\begin{cases}
    attacks(last), & \text{if $turn = Opp$ and} \\ & \text{\textit{attacks}(\textit{last})$\ne \emptyset$} \\
    \{nothing\} & \text{otherwise}
\end{cases}
\] where \textit{attacks(y)} = $\{attack_{xy} \mid attack_{xy} \in Act_{Opp}\}$.

Unlike the standard agents, \textit{Env} only has one
action available:
$P_{Env}(\_, \_, \_)$ = \{\textit{nothing}\}.

\item The transition functions determine the evolution of the debate, based on the agents' actions and their current local state:
\begin{itemize}
    \item $t_{Pro}(\textunderscore, (\textit{Pro, y, \textunderscore}), (attack_{xy}, nothing, nothing)) = \textit{x}$
    \item $t_{Opp}(\textunderscore, (\textit{Opp, y, \textunderscore}), (nothing, attack_{xy}, nothing)) = \textit{x}$
    \item $t_{Env}$((\textit{Pro, y, attacksSeen}), ($attack_{xy}, nothing, nothing$)) = (\textit{Opp, x,} $attacksSeen \cup$\{($x, y$)\})
    \item $t_{Env}$((\textit{Opp, y, attacksSeen}), ($nothing, attack_{xy}, nothing$))= (\textit{Pro, x,} $attacksSeen \cup$\{($x, y$)\})
\end{itemize}

\item For each argument $b \in Args$, we consider two propositional
  atoms, $Pro_b$ and $Opp_b$. The global states $g \in G$ of $IS_T$ at
  which these atoms hold are described by the valuation function
  $h_T$:
\begin{itemize}
    \item $h_T(Pro_x) = \{g \in G \mid l_{Pro}(g) = x$  and $l_{Env}(g)$ is a tuple of the form $(Opp, x, \_)\}$
    \item $h_T(Opp_x) = \{g \in G \mid l_{Opp}(g) = x$  and $l_{Env}(g)$ is a tuple of the form $(Pro, x, \_)\}$
\end{itemize}
 where $l_{Pro}$, $l_{Opp}$, and $l_{Env}$ are functions that return the local states of \textit{Pro}, \textit{Opp} and \textit{Env} respectively from a given global state $g \in G$.
\end{itemize}
\end{definition}


The interpreted system $IS_T$, for a dispute tree $T$, evolves as each
player selects an argument to play. The local states for proponent
\textit{Pro} and opponent \textit{Opp} therefore store their most
recently played argument. Both agents may select any argument in
the set $Args$.
The root argument of $T$, which is played by the proponent, is the
initial local state of \textit{Pro}; whereas the initial local state
of \textit{Opp} is \textit{empty}.  This state takes into
consideration the case where \textit{Opp} has not yet put forward an
argument, while \textit{Pro} has already played \textit{a}. The initial
state of \textit{Env} also indicates that \textit{a} has already been
played and \textit{Opp} must now find an argument attacking
\textit{a}.


At any given point in time, exactly one agent \textit{Pro} or
\textit{Opp} plays an argument. Thus, we use $Env$ to store
information about the agent's turn.
Moreover,
$Env$ also tracks the most recently
played argument in the dispute as well as the
attacks seen so far.



In \textit{T}, each player attacks their
counterpart’s argument made in a single dispute, with an argument of their
own.
\textit{Pro} and \textit{Opp} have actions for all attacks available
in the abstract argumentation framework $\mathcal{AF}$, which may be
enabled depending on their private local state along with
the information stored in the local state of \textit{Env}.

For each $(x, y) \in Att$, there is a corresponding
action $attack_{xy}$. The action \textit{nothing} is also available,
which an agent may select when it is their counterpart's turn or when
their counterpart's most recently played argument in the dispute is
unattacked. This is indicated in the protocol functions above, where
the only enabled action for an agent is \textit{nothing}, when it is
their counterpart's turn to play. Otherwise, the agent must find all
actions corresponding to an attack against the most recently played
argument, \textit{last}. This information is stored in the local state
of \textit{Env}, which \textit{Pro} and \textit{Opp} may observe.

%


The transition functions for \textit{Pro} and \textit{Opp} update
their private local states, to their most recently played argument, based on the action taken by all of the agents.
The states of \textit{Pro} and \textit{Opp} will not change if their chosen action is
\textit{nothing}.
Similarly, the local state of \textit{Env} is updated according to the
actions taken by the standard agents. All three variables stored by
\textit{Env} are updated simultaneously.

Finally, the valuation function above describes the states at which the
propositional atoms hold. These atoms indicate the most recently
played argument at a global state, and by which agent.



\label{sec:IST_val}

\section{Formalising Winning Conditions on Dispute Trees}
\label{Sec4}

In this section, we formalise the dispute tree properties in
Def.~\ref{def:grd_win}, \ref{def:adm_win}, and
\ref{def:ideal_win}. These formulae can be interpreted on debate
transition systems to verify the existence of a winning strategy
$\sigma$ under a particular semantics. 

Consider an abstract argumentation framework $\mathcal{AF}= \langle
Args,\allowbreak Attacks \rangle$ and an argument $a \in Args$. Let
$T$ be the dispute tree induced by $a$. By translating $T$ according
to Def.~\ref{IST}, we obtain the interpreted system $IS_T$ with
initial state $s_0$.
Further, $Var$ refers to a fixed set of variables with
$p, o, e \in Var$. Finally, the set of paths obtained by applying
strategy $f_i$, for some agent $i \in \{Pro, Opp, Env\}$, from state
$s$ in $IS_T$ is referred to as $out(s, \{f_i\})$.
\subsection{Grounded Winning Strategy} \label{theo:grd}

For verifying the existence of a grounded winning strategy $\sigma$ in
$T$ by using $IS_T$, ATL is used
to formalise the dispute tree properties in Def.~\ref{def:grd_win}.
\begin{theorem} \label{theor1}
The proponent has a grounded winning strategy $\sigma$ in $T$ iff $IS_T \models \llangle\{Pro\}\rrangle F \, \Big(\,\underset{i \in Args}{\bigvee} AG(Pro_i)\Big)$.
\end{theorem}

The formula in Theorem~\ref{theor1} states that $Pro$ has a strategy
$f$ from the initial state $s_0$ in $IS_T$, such that every path
$\lambda$ in
$out(s_0, \{f\})$ must have a steady state $s_k$ ($k \geq 0$) along
the path, where $Pro$ will play some argument $i \in Args$ that is
unattacked. We use the CTL universal quantifier $A$ instead of its ATL
equivalent $\llangle\emptyset\rrangle$ for readability.

\begin{proof}
Let $D_T$ refer to the set of disputes beginning at $a$ in the dispute
tree $T$, where each dispute is a sequence of attacking arguments. For
example, if $\delta \in D_T$ is an infinite sequence $a_0, a_1,\ldots$,
then for all $i \geq 0$, $(a_{i+1}, a_i) \in Att$. With an abuse
of notation, we write $a \in \delta$ to say that argument $a$ occurs
in $\delta$. Similarly, we write $s_i \in \lambda$, for $i \geq 0$,
to say that state $s_i$ occurs in path $\lambda$ in $IS_T$.

$\Longrightarrow$ Consider the set $D_{T_\sigma} \subseteq D_T$ of
disputes for a grounded winning strategy $\sigma$. Each dispute
$\delta \in D_{T_\sigma}$ corresponds to a path $\lambda = s_0,
s_1,\ldots$. Consequently, there exists a set $P$ of paths  from $s_0$ that
is associated with the set $D_{T_\sigma}$ of disputes.

Select an arbitrary path $\lambda \in P$. The dispute $\delta \in
D_{T_\sigma}$ corresponding to $\lambda$ must be finite as it occurs
in $T_\sigma$. This implies that there is a terminating argument
$y \in \delta$ labelled by the proponent. Accordingly, there must be a
state $s_k \in \lambda$, for $k \geq 0$, such that $(IS_T,
s_k) \models Pro_y$.

Then, select an arbitrary path $\lambda'$ from $s_k$ and a state
$s_{k'} \in \lambda'$, for $k' \geq k$. As $y$ is the winning
argument, there is no further action associated with an attack from
either agent. The dispute must have reached a steady state and
therefore, $(IS_T, s_{k'}) \models Pro_y$. This holds for all $k' \geq
k$. As this holds for an arbitrary $\lambda'$, it must hold for all
possible paths from $s_k$ and so we have $(IS_T, s_k) \models
AG(Pro_y)$. Thus, we also have $(IS_T, s_k)\models \,\underset{i \in
Args}{\bigvee} AG(Pro_i)$ (1).

As (1) holds for an arbitrarily chosen path, then for all other paths in $P$ there must exist a state $s_k$ ($k \geq 0$) where (1) holds. (2)

Finally, as the proponent has a winning strategy $\sigma$, there is a corresponding strategy $f$ for $Pro$ in $IS_T$ that allows $Pro$
to select an action deterministically at each state. $P$ is the result
of applying $f$ from the initial state $s_0$ and so we have that
$out(s_0, \{f\}) = P$ (3).
From (2) and (3), it follows that $IS_T \models \llangle\{Pro\}\rrangle F \, \Big(\,\underset{i \in Args}{\bigvee} AG(Pro_i)\Big)$.
As $\lambda$ was selected arbitrarily, every other path in
$out(s_0, \{f\})$ must also be a translation of a finite dispute
$\delta \in D_T$ (3). That is,
there is a set $D_{T_\sigma} \subseteq D_T$ of finite disputes, where
each dispute in $D_{T_\sigma}$ is translated to a path in
$out(s_0, \{f\})$. The disputes in $D_{T_\sigma}$ form a subtree
$T_\sigma$, with $a$ being the root argument, where:
\begin{itemize}
    \item Every opponent node has exactly one child. An opponent node
    always has a child as the terminating argument in each dispute in
    $T_\sigma$, labelled by the proponent. An opponent node must have no more than one child as $T_\sigma$ is
    associated with $out(s_0, \{f\})$, the set of paths obtained from
    $Pro$ always selecting one action at each state.  \item Every
    dispute in $T_\sigma$ is finite. This is shown as (2).
\end{itemize}
\end{proof}
By Theorem~\ref{theor1}, we can check whether the proponent has a
grounded winning strategy in $T$, and therefore whether the argument
at the root of $T$ is acceptable according to the grounded semantics,
by model checking the corresponding ATL formula on $IS_T$.


\subsection{Admissible Winning Strategy}
\label{theo:adm}
For verifying the existence of an admissible winning strategy $\sigma$
using $IS_T$, SL is used to formalise the dispute tree
properties in Def.~\ref{def:adm_win}.
\begin{theorem} \label{theorem2}
The proponent has an admissible winning strategy $\sigma$ in $T$ iff
$IS_T \models \exists p \, \forall
e \,\Big(\phi_1 \bigwedge \phi_2 \Big)$ where
\begin{align*}
\phi_1 &=  \forall o \,G\,\Big(\bigwedge_{i \in Args} \neg G(Opp_i)\Big)\\
\phi_2 &=  \bigwedge_{i \in Args} \Big( \exists o\, F(Pro_i) \Rightarrow \,\forall o \, G(\neg Opp_i) \Big)
\end{align*}
\end{theorem}
\begin{proof}
Note that $Pro(\delta)$ refers to the set of arguments played by the proponent in a dispute $\delta$.
\begin{description}[leftmargin=0pt]
\item Let $Str_{Pro}$, $Str_{Opp}$ be the set of strategies available to $Pro$ and $Opp$.
\item$\Longrightarrow$ Agent $Pro$ must have a strategy $f_{Pro} \in Str_{Pro}$ such that every path $\lambda = s_0, s_1,..$ in the set $out(s_0, \{f_{Pro}\})$ corresponds to a dispute $\delta \in D_{T_\sigma}$, where $s_0$ is the initial state of $IS_T$. (1)
\item The environment $Env$ only has one action available, which does not affect the evolution of the debate. This means that $Env$ only has one strategy, $f_{Env}$, where it always selects the action $nothing$. As $f_{Env}$ does not impact the evolution of the debate, the set of paths obtained from $Pro$ and $Env$ applying their strategies $f_{Pro}$ and $f_{Env}$ from $s_0$, must still be $out(s_0, \{f_{Pro}\}))$. (2)
\item Using $Var$, we define a new assignment $\chi_0$ where $\chi_0(p)$ = $f_{Pro}$ and $\chi_0(e)$ = $f_{Env}$.
\item We select an arbitrary strategy $f_{Opp} \in Str_{Opp}$ for $Opp$, map variable $o$ to $f_{Opp}$ ($\chi_1$ = $\chi_0[o \mapsto f_{Opp}]$). We now have a \textit{complete} assignment, $\chi_1$, as every agent has been assigned a strategy.
\item From $\chi_1$, we obtain the ($\chi_1$-$s_0$) play $\lambda \in out(s_0, \{f_{Pro}\})$. Then, we select a state $s_k \in \lambda$, where $k \geq 0$.
\item For each $i \in Args$:
\begin{itemize}
    \item Assume that $(IS_{T},\chi_1, s_k) \models G(Opp_i)$. This means that for all states $s_{k'} \in \lambda'$ ($k' \geq k$), where $\lambda'$ is the $\chi_1$-$s_k$ play, it holds that $(IS_{T},\chi_1, s_{k'}) \models Opp_i$. This implies that $Pro$ has no further actions available that could correspond to an attack against $i$. This suggests that the dispute $\delta \in D_{T_\sigma}$ corresponding to $\lambda$ contains a terminating argument $i$ that is labelled by the opponent (3).
    \item (3) contradicts the initial assumption that there exists an admissible winning strategy $\sigma$, as $\delta$ contains an opponent node $i$ that does not have a child.
    \item Therefore, we have that $(IS_{T},\chi_1, s_k) \models \neg G(Opp_i)$.
\end{itemize}
\item We have that $(IS_{T},\chi_1, s_k) \models \underset{i \in Args}{\bigwedge}\neg G(Opp_i)$. This must hold for all $s_k \in \lambda$ ($k \geq 0$) so that the initial assumption, that the proponent has an admissible winning strategy $\sigma$ is not contradicted. As a result, we also have that $(IS_{T},\chi_1, s_0) \models G\,\Big(\,\underset{i \in Args}{\bigwedge}\ \neg G(Opp_i)\Big)$.
\item The strategy $f_{Opp}$ was chosen arbitrarily. It must be the case that for all strategies $f_{Opp}' \in Str_{Opp}$, it holds that ($IS_T,\chi_0,[o \mapsto f_{Opp}'], s_0) \models G\,\Big(\underset{i \in Args}{\bigwedge}\ \neg G(Opp_i)\Big)$.
\item It follows that $(IS_T,\chi_0, s_0) \models \forall o \,G\,\Big(\,\underset{i \in Args}{\bigwedge}\ \neg G(Opp_i)\Big)\Big)$. (4)
\item We have assignment $\chi_0$. For each $i \in Args$:
\begin{itemize}
    \item Assume that $(IS_{T},\chi_0,s_0) \models \exists o \,F(Pro_i)$. This means that $Opp$ must have some strategy $f_{Opp} \in Str_{Opp}$ that allows this assumption to hold. We associate $Opp$'s strategy $f_{Opp} \in Str_{Opp}$ to variable $o$ and obtain an assignment $\chi_1$ ($\chi_1$ = $\chi_0[o \mapsto f_{Opp}]$). Let $\lambda \in out(s_0, \{f_{Pro}\}))$ be the $\chi_1$-$s_0$ play and let $\delta \in D_{T_\sigma}$ be the dispute corresponding to $\lambda$.
    \item There exists a state $s_k \in \lambda$ ($k \geq 0$) where it holds that $(IS_{T},\chi_2, s_k) \models Pro_i$. Accordingly, we have that the proponent plays argument $i$ in $\delta$ at some stage.
    \item There must be no dispute $\delta' \in D_{T_\sigma}$ where $i$ is labelled by the opponent, as $\sigma$ is admissible. (5)
    \item For each strategy $f_{Opp}' \in Str_{Opp}$ available to $Opp$, we obtain a new assignment $\chi_1$ from $\chi_0$ ($\chi_1$ = $\chi_0[o \mapsto f_{Opp}']$). According to (5), for the $\chi_1$-$s_0$ play $\lambda' \in  out(s_0, \{f_{Pro}\})$, there must not exist a state $s_{k'} \in \lambda$ ($k' \geq 0$) where it holds that $(IS_{T},\chi_1 ,s_{k'}) \models Opp_i$. This means that $(IS_{T}, \chi_1, s_0) \models G(\neg Opp_i)$. (6)
    \item As (6) holds for each strategy $f_{Opp}'$ available to $Opp$, it must be the case that $(IS_{T}, \chi_0 ,s_0) \models \forall o \,G(\neg Opp_i)$.
\end{itemize}
\item We have $(IS_{T},\chi_0, s_0) \models\underset{i \in Args}{\bigwedge}\exists o \, F(Pro_i) \Rightarrow \forall o \,G(\neg Opp_i)$. (7)
\item From (4) and (7), we have that $(IS_{T},\chi_0, s_0) \models \Big(\phi_1 \bigwedge \phi_2 \Big)$. (*)
\item From (1), (2), and (*) it follows that $IS_T \models \exists p \, \forall e\, \Big(\phi_1 \bigwedge \phi_2\Big)$.

\item$\Longleftarrow$ $Pro$ has some strategy $f_{Pro} \in Str_{Pro}$ and $Env$ has its only strategy $f_{Env}$, which does not affect the evolution of the system. We obtain a set of paths $out(s_0, \{f_{Pro}\})$ from the initial state $s_0$, when both $Pro$ and $Env$ apply their strategies. The set $out(s_0, \{f_{Pro}\})$ corresponds to a set of disputes $D_{T_\sigma} \subseteq D_T$.
\item We have an assignment $\chi_0$ where $\chi_0(p)$ = $f_{Pro}$ and $\chi_0(e)$ = $f_{Env}$. We have that: 
\begin{align*}
&(IS_{T},\chi_0 ,s_0) \models \forall o \, G\,\Big(\bigwedge_{i \in Args} \neg G(Opp_i)\Big)\Big) &(1)\\
&(IS_{T},\chi_0,s_0) \models \Big( \bigwedge_{i \in Args} \exists o \, F(Pro_i) \Rightarrow \, \forall o \,G(\neg Opp_i)\Big)\Big) &(2)
\end{align*}
\item Select an arbitrary strategy $f_{Opp} \in Str_{Opp}$. We obtain a new complete assignment $\chi_1$ = $\chi_0[o \mapsto f_{Opp}]$.
We now have a ($\chi_1$-$s_0$) play $\lambda \in out(s_0, \{f_{Pro}\})$. We have that $(IS_{T}, \chi_1, s_k) \models \underset{i \in Args}{\bigwedge} \neg G(Opp_i)$, for all $s_k \in \lambda$ ($k \geq 0$). This implies that for all $i \in Args$ there exists at least one state $s_{k'} \in \lambda'$, where $\lambda'$ is the $\chi_1$-$s_k$ play and $k' \geq k$, such that $(IS_{T},\chi_1, s_{k'}) \not \models Opp_i$. This suggests that $Pro$ has been able to counter attack any argument $i$ that was put forward by $Opp$ in $\lambda$. Accordingly, the dispute $\delta \in D_{T_\sigma}$ that corresponds to $\lambda$ does not contain opponent nodes with zero children. (3)
\item From (1), we know that (3) must hold for all strategies available to $Opp$. Therefore, for all plays $\lambda \in out(s_0, \{f_{Pro}\})$, their corresponding disputes $\delta \in D_{T_\sigma}$ do not contain opponent nodes with zero children (*).
\item We have assignment $\chi_0$. Select an arbitrary dispute $\delta \in  D_{T_\sigma}$. The play $\lambda \in out(s_0, \{f_{Pro}\})$ that corresponds to $\delta$ must be obtained from some complete assignment $\chi_1$ = $\chi_1[o \mapsto f_{Opp}]$, where $f_{Opp} \in Str_{Opp}$ is some strategy of $Opp$ (4).  
\item For some argument $b \in Pro(\delta)$, there must exist a state $s_k \in \lambda$ ($k \geq 0$) such that $(IS_{T},\chi_1, s_k) \models Pro_b$. This gives us $(IS_{T}, \chi_1, s_0) \models F(Pro_b)$. Given (4), we have that $(IS_{T},\chi_0,s_0) \models \exists o \, F(Pro_i)$ (5).
\item From (2) and (5), we have that $(IS_{T},\chi_0,s_0) \models \forall o \,G(\neg Opp_b)$ (6).
\item For some strategy $f_{Opp}' \in Str_{Opp}$ available to $Opp$, we obtain a new complete assignment $\chi_1$ = $\chi_0[o \mapsto f_{Opp}']$. We know that for all states $s_k$ ($k \geq 0$) in the $\chi_1$-$s_0$ play $\lambda' \in out(s_0, \{f_{Pro}\})$, it must hold that $(IS_{T},\chi_1, s_k) \models \neg Opp_b$. Accordingly, the dispute $\delta'$ that corresponds to $\lambda'$ must not have $b$ labelled by the opponent. From (6), we know that for every strategy $f_{Opp}' \in Str_{Opp}$ available to $Opp$, the $\chi_1$-$s_0$ play $\lambda'\in out(s_0, \{f_{Pro}\})$ corresponds to a dispute $\delta' \in D_{T_\sigma}$ where $b$ has not been labelled by the opponent. This means that there is no dispute $\delta' \in D_{T_\sigma}$, where the opponent labels $b$ (7).
\item As $b$ was arbitrary, (7) must hold for every argument $b \in Pro(\delta)$ (8).
\item As $\delta$ was also selected arbitrarily, (8) must hold for every dispute $\delta \in D_{T_\sigma}$ (**).
\item We have that:
\begin{itemize}
    \item Every opponent node in $T_\sigma$ has children. The number of children must be exactly one as $T_\sigma$ is associated with $out(s_0, \{f_{Pro}\})$, the set of paths obtained from $Pro$ always selecting one action at each state. This was shown in (*).
    \item No argument is labelled by both the proponent and the opponent. This was shown in (**).
\end{itemize}
Hence, $\sigma$ must be an admissible winning strategy in $T$.
\end{description}
\end{proof}

Intuitively, the formula in Theorem~\ref{theorem2} states
that \textit{Pro} has a strategy $f$ in $IS_T$ such that
\begin{enumerate}
    \item $\phi_1$ is satisfied. The subformula $\phi_1$ states that for any strategy $f'$ applied by \textit{Opp}, the resulting
    play $\lambda$ does not contain a steady state $s_k$ ($k \geq
    0$) that is associated with an argument played by $Opp$. Thus, for any $i \in Args$, $\lambda$ must not contain a
    steady state $s_k$ where a propositional atom $Opp_i$
    holds.

\item $\phi_2$ is satisfied. The subformula $\phi_2$ states that for any two plays $\lambda, \lambda' \in
    out(s_0, \{f\})$, no argument $i \in Args$ must be put forward by
    both agents. As a result, if $Pro_i$ holds at some state
    $s_k \in \lambda$ ($k \geq 0$), then there must be no state
    $s_{k'} \in \lambda'$ ($k' \geq 0$) where $Opp_i$ holds.
\end{enumerate}


\subsection{Ideal Winning Strategy}
\label{theo:ideal}
For verifying the existence of an ideal winning strategy $\sigma$
using $IS_T$, SL is used to formalise the dispute tree
properties in Def.~\ref{def:ideal_win}.
\begin{theorem} \label{theorem3}
The proponent has an ideal winning strategy $\sigma$ in $T$ iff
$IS_T \models \exists p\, \forall
e \Big(\phi_{1} \bigwedge \phi_{2} \bigwedge \phi_3 \Big)$, where
$\phi_1$ and $\phi_2$ are the same as in the statement of
Theorem~\ref{theorem2}, and 
\begin{align*}
&\phi_{3} = \forall o \,G\Big(\Big(\underset{i \in Args}{\bigvee} Opp_i \Big) \Rightarrow \neg \exists o \, \Big(\phi_{3.1}\, \mathsmaller{\mathsmaller{\bigwedge}} \, \phi_{3.2} \Big)\Big)
\end{align*} where
\begin{align*}
&\phi_{3.1} = \forall p \,G\,\Big(\bigwedge_{j \in Args} \neg G(Pro_j)\Big)\\
&\phi_{3.2} = \Big( \bigwedge_{j \in Args} \exists p \, F(Pro_j) \Rightarrow \forall p \, G(\neg Opp_j)\Big)
\end{align*}
\end{theorem}
\begin{proof}
$Pro(\delta)$ refers to the set of arguments played by the proponent in a dispute $\delta$, while $Opp(\delta)$ refers to the set of arguments played by the opponent in $\delta$. $Opp(T_\sigma)$ refers to the set of arguments played by the opponent in the tree $T_\sigma$.
\begin{description}[leftmargin=0pt]
\item Let $Str_{Pro}$, $Str_{Opp}$ be the set of strategies available to $Pro$ and $Opp$.
\item $\Longrightarrow$ Agent $Pro$ must have a strategy $f_{Pro} \in Str_{Pro}$ such that every path $\lambda = s_0, s_1,..$ in the set $out(s_0, \{f_{Pro}\})$ corresponds to a dispute $\delta \in D_{T_\sigma}$, where $s_0$ is the initial state of $IS_T$. (1)
\item The environment $Env$ only has one action available, which does not affect the evolution of the debate. This means that $Env$ only has one strategy, $f_{Env}$, where it always selects the action $nothing$. As $f_{Env}$ does not impact the evolution of the debate, the set of paths obtained from $Pro$ and $Env$ applying their strategies $f_{Pro}$ and $f_{Env}$ from $s_0$, must still be $out(s_0, \{f_{Pro}\})$. (2)
\item Using $Var$, we define a new assignment $\chi_0$ where $\chi_0(p)$ = $f_{Pro}$ and $\chi_0(e)$ = $f_{Env}$.
\item As $\sigma$ must be admissible, we have that:
\begin{align*}
&(IS_T,\chi_0 ,s_0) \models \phi_1\,\,&(3) \\
&(IS_T,\chi_0, s_0) \models \phi_2\,\, &(4)
\end{align*}
Please see the proof of Theorem \ref{theo:adm} for (3) and (4).
\item We have assignment $\chi_0$. Select an arbitrary strategy $f_{Opp} \in Str_{Opp}$.  We associate $f_{Opp}$ to variable $o$ such that we obtain assignment $\chi_1$ = $\chi_0[o \mapsto f_{Opp}]$. As every agent has now been assigned a strategy, we have an $\chi_1$-$s_0$ play $\lambda \in out(s_0, \{f_{Pro}\})$. Select an arbitrary state $s_k \in \lambda$ ($k \geq 0$).
\item  Assume that $(IS_{T},\chi_1,s_k) \models  \underset{i \in Args}{\bigvee} Opp_i$. 
\item Now, assume that $(IS_{T},\chi_1, s_k) \models Opp_i$, where $i$ is an argument in $Args$ (5).
\item Also, assume that $(IS_T, \chi_1, s_k) \models \exists o \, \Big(\phi_{3.1} \bigwedge \phi_{3.2} \Big)$ (6). 
\item From (6), we know that $Opp$ has some strategy $f_{Opp}' \in Str_{Opp}$. We now obtain a new assignment $\chi_2$ = $\chi_1[o \mapsto f_{Opp}']$. It holds that $(IS_{T},\chi_2,s_k) \models \phi_{3.1} \bigwedge \phi_{3.2}$ (7).
\item Let $out(s_k, \{f_{Opp}'\})$ be the set of paths from $s_k$, obtained from $Opp$ applying its strategy $f_{Opp}'$. The set $out(s_k, \{f_{Opp}'\})$ corresponds to a subtree $T_{\sigma'}$ of $T$ where $i$ is the root node labelled by the opponent and $\sigma'$ is a strategy of the opponent from $i$. We select some strategy $f_{Pro} \in Str_{Pro}$ for $Pro$ and assign it to the variable $p$, such that $\chi_3$  = $\chi_2[p \mapsto f_{Pro}]$.
From (7), we now have that $(IS_{T}, \chi_3, s_k) \models \phi_{3.1}$ (8).
\item Let $\lambda' \in out(s_k, \{f_{Opp}'\})$ be the $\chi_3$-$s_k$ play. According to (8), for all $s_{k'} \in \lambda'$ ($k' \geq k$), it holds that $(IS_{T},\chi_3, s_{k'}) \models \underset{j \in Args}{\bigwedge} \neg G(Pro_j)$. This implies that $Opp$ is able to counter attack any argument played by $Pro$ in $\lambda'$. Hence, the dispute $\delta' \in D_{T_{\sigma'}}$ that is associated with $\lambda'$ does not contain proponent nodes with zero children (9).
\item According to (8), (9) must hold for any strategy $f_{Pro} \in Str_{Pro}$ available to $Pro$. Therefore, all plays $\lambda' \in out(s_k, \{f_{Opp}'\})$ correspond to a dispute $\delta' \in D_{T_{\sigma'}}$, where there is no proponent node with zero children. The number of children for each proponent node must be exactly one as $T_{\sigma'}$ is associated with $out(s_k, \{f_{Opp}\})$, the set of paths obtained from $Opp$ always selecting one action at each state. (*)
\item We have assignment $\chi_2$. Select an arbitrary dispute $\delta' \in  D_{T_{\sigma'}}$. The play $\lambda' \in out(s_0, \{f_{Opp}'\})$ that corresponds to $\delta'$ must be obtained from some assignment $\chi_3$ = $\chi_2[p \mapsto f_{Pro}])$, where $f_{Pro} \in Str_{Pro}$ is some strategy of $Pro$ (10).
\item The assignment $\chi_3$ is \textit{complete}, as every agent has been assigned a strategy.
\item For some argument $b \in Pro(\delta')$, there must exist a state $s_{k'} \in \lambda'$ ($k' \geq k$) where $(IS_{T},\chi_3,s_{k'}) \models Pro_b$. This gives us $(IS_{T},\chi_3, s_k) \models F(Pro_i)$. 
\item From (10), we have $(IS_{T},\chi_2, s_k) \models \exists p \, F(Pro_b)$ (11).
\item From (11) and (7), we obtain $(IS_{T},\chi_2, s_k) \models  \forall p \, G(\neg Opp_b)$ (12). 
\item For some strategy $f_{Pro}' \in Str_{Pro}$ available to $Pro$, we obtain a new complete assignment $\chi_3$ = $\chi_2[p \mapsto f_{Pro}']$.
We know that for all states $s_{k'}$ ($k' \geq k$) in the $\chi_3$-$s_k$ play $\lambda'' \in out(s_k, \{f_{Opp}'\})$, it must hold that $(IS_{T},\chi_3, s_{k'}) \models \neg Opp_b$. Accordingly, the dispute $\delta'' \in D_{T_{\sigma'}}$ that corresponds to $\lambda''$ must not have $b$ labelled by the opponent. From (12), we know that for every strategy $f_{Pro}' \in Str_{Pro}$, the $\chi_7$-$s_k$ play $\lambda''\in out(s_k, \{f_{Opp}'\})$ corresponds to a dispute $\delta'' \in D_{T_{\sigma'}}$ where $b$ has not been labelled by the opponent. This means that there is no dispute $\delta'' \in D_{T_{\sigma'}}$, where the opponent labels $b$ (13).
\item (13) must hold for every argument $b \in Pro(\delta')$ (14).
\item As $\delta$ was also selected arbitrarily, (13) must hold for every dispute $\delta \in D_{T_{\sigma'}}$ (**). From (*) and (**), it follows that $\sigma'$ is an admissible winning strategy for an opponent node $i$ (15).
\item (15) contradicts the initial assumption that $\sigma$ is ideal as $i$ has an admissible winning strategy $\sigma'$. Therefore, the assumption made in (6) does not hold and we have that $ (IS_T$,$\chi_1, s_k) \models \neg \exists o \, \Big(\phi_{3.1} \bigwedge \phi_{3.2} \Big)$ (16).
\item The assumption in (5) can be made about any argument $i \in Args$ and would still lead to (16).
\item We have that $(IS_{T},\chi_1,s_k) \models \Big(\underset{i \in Args}{\bigvee}Opp_i\Big) \Rightarrow \neg \exists o \, \Big(\phi_{3.1} \bigwedge \phi_{3.2} \Big)$  (17).
\item (17) must hold for all $s_k \in \lambda$, where $k \geq 0$, as $s_k$ was selected arbitrarily. As a result, we have that $(IS_{T},\chi_1, s_0) \models G\,\Big(\Big(\underset{i \in Args}{\bigvee} Opp_i\Big) \Rightarrow \neg \exists o \, \Big(\phi_{3.1} \bigwedge \phi_{3.2} \Big)\Big)$ (18).
\item This implies that there is no admissible winning strategy for any of the nodes labelled by the opponent in $\delta \in D_{T_\sigma}$, the dispute that is associated with $\lambda$ (19). This means that for all other strategies $f_{Opp} \in Str_{Opp}$, (19) must hold.
\item We have $(IS_{T},\chi_0,s_0)\models \llbracket o \rrbracket (Opp, o)\,G\Big(\Big(\underset{i \in Args}{\bigvee} Opp_i\Big) \Rightarrow \neg \phi_{2}\Big)$ (20).
\item We have (3), (4), and (20). From (1) and (2), it follows that:
\begin{align*}
IS_T \models \exists p \, \forall e \, \Big(\phi_{1} \bigwedge \phi_{2} \bigwedge \phi_{3} \Big). 
\end{align*}

\item $\Longleftarrow$ $Pro$ has some strategy $f_{Pro} \in Str_{Pro}$ and $Env$ has its only strategy $f_{Env}$, which does not affect the evolution of the system. We obtain a set of paths $out(s_0, \{f_{Pro}\})$ from the initial state $s_0$, when both $Pro$ and $Env$ apply their strategies. The set $out(s_0, \{f_{Pro}\})$ corresponds to a set of disputes $D_{T_\sigma} \subseteq D_T$.
\item We have an assignment $\chi_0$ where $\chi_0(p)$ = $f_{Pro}$ and $\chi_0(e)$ = $f_{Env}$.
We have that
\begin{align*}
&(IS_{T},\chi_0, s_0) \models \phi_{1} \wedge \phi_{2}\,\, &(1)\\ &(IS_{T},\chi_0, s_0) \models \forall o \,G\Big(\Big(\underset{i \in Args}{\bigvee} Opp_i \Big) \Rightarrow \neg \exists o \, \Big(\phi_{3.1} \bigwedge \phi_{3.2} \Big)\Big)\,\, &(2)
\end{align*}
\item Let $T_\sigma$ be the dispute tree, with root $a$, constructed from the set of disputes $D_{T_\sigma}$. The proof of Theorem \ref{theo:adm} shows that if (1) holds, then $\sigma$ must be admissible (*). 
\item Consider $Opp(T_\sigma)$. There are two cases: $Opp(T_\sigma) \neq \emptyset$ (Case 1) and $Opp(T_\sigma) = \emptyset$ (Case 2).  
\item (Case 1) Select an arbitrary dispute $\delta \in D_{T_\sigma}$. The play $\lambda \in out(s_0, \{f_{Pro}\})$ that corresponds to $\delta$ must be obtained from some complete assignment $\chi_1$ = $\chi_0[o \mapsto f_{Opp}]$, where $f_{Opp}$ is some strategy available to $Opp$.
\item We select some argument $b \in Opp(\delta)$ arbitrarily. There must be a state $s_k \in \lambda$ ($k \geq 0$) such that $(IS_{T}, \chi_1, s_k) \models Opp_b$ (3). 
\item From (3), it holds that $(IS_{T}, \chi_1, s_k) \models \bigvee_{i \in Args} Opp_i$ (4).
\item From (4) and (2), we have that $(IS_{T}, \chi_1, s_k) \models \neg \phi_2$.
\item Consider a strategy $f_{Opp}' \in Str_{Opp}$. The set $out(s_k, \{f_{Opp}'\})$ corresponds to a set of disputes $D_{T_{\sigma'}}$, which can be used to construct a dispute tree $T_{\sigma'}$ where $b$ is the root and is labelled by the opponent. We now show that the opponent's strategy $\sigma'$ is not an admissible winning strategy for $b$.
\item Let $\chi_2$ = $\chi_1[o \mapsto f_{Opp}'$].
We have $(IS_{T},\,\chi_2, s_k) \not \models \phi_{3.1} \bigwedge \phi_{3.2}$.
\item This implies that one of two, or both, could hold:
\begin{align*}
&(IS_{T},\chi_2,s_k) \not \models \phi_{3.1} &(5)\\  &(IS_{T},\chi_2 ,s_k) \not \models \phi_{3.2} &(6)
\end{align*}
\item Consider the first case, which is the statement in (5). There must exist some strategy $f_{Pro}' \in Str_{Pro}$ such that 
\item $(IS_{T}, \chi_3, s_k) \not \models G\Big(\bigwedge_{j \in Args} \neg G(Pro_j)\Big)$, where $\chi_3$ = $\chi_2[p \mapsto f_{Pro}'$].
There must exist some state $s_{k'} \in \lambda'$ ($k' \geq k$), where $\lambda' \in out(s_k, \{f_{Opp}'\})$ is the $\chi_3$-$s_k$ play, such that $(IS_{T}, \chi_3, s_{k'}) \not \models \bigwedge_{j \in Args} \neg G(Pro_j)$. This implies that for some argument $j \in Args$, we have that for all $s_{k''} \in \lambda'$ ($k'' \geq k'$), $(IS_{T},\chi_3,s_{k''}) \models Pro_j$. This shows that there are no attacks available against $j$. 
\item Accordingly, there exists a dispute $\delta' \in D_{T_{\sigma'}}$, that corresponds to $\lambda'$, where $j$ is the terminating argument labelled by the proponent. As $j$ has no children, $\sigma'$ cannot be an admissible winning strategy of the opponent, for $b$ (**).
\item Consider (6). There must be some argument $j \in Args$ such that \newline $(IS_{T},\chi_2, s_k) \not \models \exists p \, F(Pro_j) \Rightarrow \forall p \, G(\neg Opp_j)$. This means that we have the following:
\begin{align*}
&(IS_{T},\chi_2,s_k) \models \exists p \, F(Pro_j)&(7) \\
&(IS_{T},\chi_2,s_k) \not \models  \forall p \, G(\neg Opp_j) &(8)
\end{align*}
\item (7) implies that there exists a strategy $f_{Pro}' \in Str_{Pro}$ such that $(IS_{T},\chi_3, s_k) \models F(Pro_j)$, where $\chi_3$ = $\chi_2[p \mapsto f_{Pro}'$].
In the $\chi_3$-$s_k$ play $\lambda' \in out(s_k, \{f_{Opp}'\})$, there must exist a state $s_{k'}$ ($k' \geq k$) such that $(IS_{T},\chi_3, s_{k'}) \models Pro_j$. Accordingly, the dispute $\delta' \in D_{T_{\sigma'}}$ that corresponds to $\lambda'$ contains an argument $j$ that is labelled by the proponent (9).
\item (8) implies that there exists a strategy $f_{Pro}'' \in Str_{Pro}$ such that $(IS_{T},\chi_3, s_k) \not \models G(\neg Opp_j)$, where $\chi_3$ = $\chi_2[p \mapsto f_{Pro}''$].
\item Therefore, there must exist a state $s_{k''} \in \lambda''$ ($k'' \geq k$), where $\lambda'' \in out(s_k, \{f_{Opp}'\})$ is a $\chi_3$-$s_k$ play, such that $(IS_{T},\chi_3,s_{k''}) \models Opp_j$. \item Accordingly, the dispute $\delta'' \in D_{T_{\sigma'}}$ that corresponds to $\lambda''$ contains an argument $j$ that is labelled by the opponent (10).
\item From (9) and (10), it follows that $\sigma'$ cannot be admissible as the same argument is labelled by both the proponent and the opponent  (***).
\item As the strategy $f_{Opp}'$ was selected arbitrarily, it must be the case that (**) and (***) hold for any $f_{Opp}' \in Str_{Opp}$ (12).
\item Similarly, $b \in Opp(\delta)$ was selected arbitrarily. This means that (12) must hold for any argument $i \in Opp(\delta)$ (13).
\item Finally, (13) must hold for any dispute $\delta \in D_{T_\sigma}$ and this implies that for any argument $i \in Opp(T_\sigma)$, the opponent does not have an admissible winning strategy for $i$ (14).
\item Consider the second case, where $Opp(T_\sigma) = \emptyset$. In this case, we find that the root argument $a$ of $T$ (and also $T_\sigma$) is unattacked. There are no opponent nodes that could have an admissible winning strategy (15).
\item We have (*), and from (14) and (15) it follows that $\sigma$ must be an ideal winning strategy for $a$ in $T$.
\end{description}
\end{proof}

Intuitively, the formula above states that \textit{Pro} has a
strategy $f$ in $IS_T$ such that
\begin{enumerate}
    \item $\phi_{1}$ and $\phi_{2}$ are satisfied. These two formulas
    are presented in Theorem \ref{theo:adm}. In particular, $f$ is
    guaranteed to be an admissible winning strategy.
    
   \item $\phi_3$ is satisfied. The subformula $\phi_3$ is used to verify that \textit{Opp} does not have a strategy $f'$ from any state $s_k \in \lambda$ ($k \geq 0$) that satisfies $Opp_i$, where $\lambda \in out(s_0, \{f\})$ and $i \in Args$, that corresponds to an admissible winning strategy for $i$.
\end{enumerate}

{\bf Discussion.} In
Theorem~\ref{theor1}, we are able to express the
existence of a grounded winning strategy in ATL. The problem of model checking an ATL formula is \textit{P-complete} \cite{Alur:2002:ATL:585265.585270}, while the problem of deciding membership of an argument to a grounded extension is in \textit{P} \cite{Dunne2009}. This suggests that it is possible to express the acceptability conditions in Def. \ref{def:grd_win} using ATL. On the other hand, we require the
strictly more expressive SL for admissible and ideal winning
strategies as deciding membership under both of these semantics is a significantly harder problem \cite{Dunne2009, ideal_complexity}. In this contribution, we do not explore whether the latter
notions can also be formalised in ATL as such a problem would require a
substantial amount of work, possibly an impossibility result, which is
beyond the scope of the present paper. We only observe that, since formulas
$\phi_2$ an $\phi_{3.2}$ refer to counterfactual situations, it is
unlikely that these can be expressed in ATL, where counterfactuals are
not readily expressible.

\section{Experimental Results}
 \label{Sec5}
\begin{table*}
    \centering
    \begin{adjustbox}{width={\textwidth},totalheight={\textheight},keepaspectratio}
    \begin{tabular}{|c|c|c|c|c|c|}
\hline
\multicolumn{2}{|c|}{\hfil\textbf{Argumentation Framework}} &  \multicolumn{3}{c|}{\textbf{MCMAS}}  &  \multicolumn{1}{c|}{\textbf{CoQuiAAS}} \\
 \hline
\textbf{Args (\textit{n})} & \textbf{Prob. attack (\textit{p})} & \textbf{Avg exec. time (s)} &\textbf{Avg reachability (s)} & \textbf{Model checking (s)} & \textbf{Avg exec. time (s)}\\ 
 \hline
\multirow{2}{2cm}{20} & $0.0 \leq p < 0.5$ &0.1742 &0.1190 &0.0131 &0.0046\\  \cline{2-6} 
 & $0.5 \leq p < 1.0$ &1.4473 &1.1280 &0.3013 &0.0048\\
 \hline
 \multirow{2}{2cm}{40} & $0.0 \leq p < 0.5$ &3.6242 &2.2740 &0.2493 &0.0049\\ \cline{2-6} 
 & $0.5 \leq p < 1.0$ &11.581 &7.7710 &3.7900 &0.0050\\
 \hline
 \multirow{2}{2cm}{60} & $0.0 \leq p < 0.5$ &617.72 &586.54 &30.291 &0.0050\\  \cline{2-6} 
 & $0.5 \leq p < 1.0$ &1532.4 &1471.3 &59.229 &0.0055\\  
 \hline
 \multirow{2}{2cm}{80} & $0.0 \leq p < 0.5$ &1683.1 &1583.6 &98.175 &0.0057\\ \cline{2-6} 
 & $0.5\leq p < 1.0$ &Timed out &Timed out &Timed out &0.0062\\ 
 \hline
\end{tabular}
\end{adjustbox}
    \caption{The performance of MCMAS compared to the performance of CoQuiAAS}
    \label{tab:grd_table}
\end{table*}

\begin{table*}
  \begin{adjustbox}{width={\textwidth},totalheight={\textheight},keepaspectratio}
    \begin{tabular}{|c|c|c|c|c|c|}
\hline
\multicolumn{2}{|c|}{\hfil\textbf{Argumentation Framework}} &  \multicolumn{3}{c|}{\textbf{MCMAS-SLK}}  &  \multicolumn{1}{c|}{\textbf{ArgSemSAT}} \\
 \hline
\textbf{Args (\textit{n})} & \textbf{Prob. attack (\textit{p})} & \textbf{Avg exec. time (s)} &\textbf{Avg reachability (s)} & \textbf{Model checking (s)} & \textbf{Avg exec. time (s)}\\ 
 \hline
\multirow{2}{2cm}{6} & $0.0 \leq p < 0.5$ &0.1700 &0.0010 &0.1866 &0.0055\\  \cline{2-6} 
 & $0.5 \leq p < 1.0$ &122.65 &0.0010 &121.39 &0.0067\\
 \hline
 \multirow{2}{2cm}{8} & $0.0 \leq p < 0.5$ &0.3470 &0.0359 &0.2493 &0.0062\\ \cline{2-6} 
 & $0.5 \leq p < 1.0$ &274.88 &0.1860 &270.42 &0.0065\\
 \hline
 \multirow{2}{2cm}{10} & $0.0 \leq p < 0.5$ &Memory error &0.1803 &Memory error &0.0071\\  \cline{2-6} 
 & $0.5 \leq p < 1.0$ &Memory error &0.5306 &Memory error &0.0076\\  \cline{2-6} 
 \hline
\end{tabular}
\end{adjustbox}
    \caption{The performance of MCMAS-SLK compared to the performance
      of ArgSemSAT}
    \label{tab:adm_table}
\end{table*}

In this section we compare the performance of MCMAS and MCMAS-SLK
\cite{CermakP2014MAmc}, a model checker that supports a fragment of
SL, to the performance of argumentation solvers in
determining the acceptability of a given argument.

The International Competition on Computational Models of Argumentation
(ICCMA) 2017 \cite{iccma} focused on several tasks involving abstract
argumentation frameworks, one of which included determining whether
some argument is acceptable under a given semantics.
For
each of the semantics, the solvers submitted were ranked according to
their performance in all tasks. The top ranking solvers for the
grounded, preferred, and ideal semantics were the open source solvers
CoQuiAAS \cite{coq2}, ArgSemSAT \cite{cerutti2014argsemsat,
argsem2017}, and Pyglaf \cite{Alviano2017ThePA} respectively.

In our experiments we used AFBenchGen2 \cite{cerutti2016generating},
an open source generator of random abstract argumentation frameworks.
Attacks are selected randomly between any two arguments $a$ and $b$
\cite{erdds1959random} by providing the generator a probability $p$ ($0 \leq p \leq 1$), which determines the likelihood of an attack from $a$ to $b$ and vice versa. 
%
For each generated framework $\mathcal{AF}$, we selected an argument
$a \in Args$ randomly and determined its acceptability. 

All tests were ran on
a machine running Linux kernel version 4.15.0-50, with 16GB RAM and an
Intel Core i7-6700 3.40GHz Processor CPU. A timeout of 1800s was set
for each test.
\subsection{Grounded Semantics}
\label{sec:grd_res}
This section compares the performance of MCMAS to that of
CoQuiAAS in determining membership of some argument in the grounded
extension. The results are presented in Table \ref{tab:grd_table}.

Using AFBenchGen2, we generated 10 argumentation frameworks per probability interval.
Each framework was automatically
translated to an interpreted
system, on which we model checked an ATL formula $\phi$
based on the argumentation framework and the formalisation in
Theorem \ref{theo:grd}.

While the results provided by MCMAS were consistent with those of
CoQuiAAS, Table \ref{tab:grd_table} shows that CoQuiAAS was
considerably faster
in determining membership
in the grounded extension. This is unsurprising as MCMAS, unlike
CoQuiAAS, is a general purpose model checker. In addition, the table
shows that on average, most of the execution time for MCMAS was spent
generating the set $G$ of reachable states, rather than performing
model checking. Computing the set of reachable states is one of the
factors that can greatly impact the performance of
MCMAS \cite{Lomuscio2017}. We observed that an increase in the number of
arguments as well as the size of the attack relation is positively correlated with the amount of time spent
by MCMAS in generating $G$. This was expected, as a larger attack relation implies that at any state, agents are
likely to have a greater number of actions enabled, resulting in an
increase in the size of $G$.
Finally, the tests that timed out from 80 arguments onward
were unable to reach the model checking step as their reachable state space was not fully generated. 

\begin{table*}
    \centering
    \begin{adjustbox}{width={\textwidth},totalheight={\textheight},keepaspectratio}
  \begin{tabular}{|c|c|c|c|c|c|}
\hline
\multicolumn{2}{|c|}{\hfil \textbf{Argumentation Framework}} &  \multicolumn{3}{c|}{ \textbf{ MCMAS-SLK}}  &  \multicolumn{1}{c|}{\textbf{ pyglaf}} \\
 \hline
 \textbf{Args (\textit{n})} & \textbf{Prob. attack (\textit{p})} & \textbf{Avg exec. time (s)} &\textbf{Avg reachability (s)} & \textbf{Model checking (s)} & \textbf{Avg exec. time (s)}\\ 
 \hline
\multirow{2}{2cm}{6} & $0.0 \leq p < 0.5$ &0.2111 &0.0010 &0.2101 &0.0570\\  \cline{2-6} 
 & $0.5 \leq p < 1.0$ &131.66 &0.0013 &130.31 &0.0600\\
 \hline
 \multirow{2}{2cm}{8} & $0.0 \leq p < 0.5$ &0.3700 &0.0210 &0.2493 &0.0590\\ \cline{2-6} 
 & $0.5 \leq p < 1.0$ &Memory error &0.1777 &Memory error &0.0633\\
 \hline
 \multirow{2}{2cm}{10} & $0.0 \leq p < 0.5$ &Memory error &0.3151 &Memory error &0.0582\\  \cline{2-6} 
 & $0.5 \leq p < 1.0$ &Memory error &0.6240 &Memory error &0.0681\\  \cline{2-6} 
\hline
\end{tabular}
\end{adjustbox}
    \caption{The performance of MCMAS-SLK compared to the performance of Pyglaf}
    \label{tab:idl_table}
\end{table*}

\subsection{Preferred Semantics}
\label{sec:pref_res}
In this section we compare the performance of MCMAS-SLK against that
of ArgSemSAT in determining acceptability of some argument under the
preferred semantics. The results are shown in
Table \ref{tab:adm_table}.

There were 10 argumentation frameworks generated per probability
interval and translated to interpreted systems, as in
Sec.~\ref{sec:grd_res}. An SL formula $\phi$, based on the
corresponding framework and the formalisation in
Theorem \ref{theo:adm}, was used to verify the existence of an
admissible winning strategy.
%
%
The results provided by MCMAS-SLK were consistent with those returned
by ArgSemSAT. As expected, ArgSemSAT was considerably faster than
MCMAS-SLK.
%
Compared to Sec.~\ref{sec:grd_res}, the argumentation frameworks
used were significantly smaller due to the higher model checking
complexity of SL\cite{CermakP2014MAmc}.
With smaller frameworks, MCMAS-SLK did not require as much time to
generate the set of reachable states. Consequently, most of the
execution time was spent model checking. For frameworks with less than
10 arguments, we found that an increase in the value of $p$ had an
significant impact on the average execution time.

Finally, we found that for most frameworks that consisted of 10 or
more arguments, MCMAS-SLK ran out of memory when verifying the
property $\phi$, rather than timing out. This may be due to the high
model checking complexity of SL as well as an increase in
the size of formula $\phi$, which is positively correlated with the
number of arguments in the framework.
\subsection{Ideal Semantics}
In this section we compare the performance of MCMAS-SLK against that
of Pyglaf in determining acceptability of some argument under the
ideal semantics.
The SL formula $\phi$ used for verifying a debate was based on the framework generated as well as the formalisation in Theorem \ref{theo:ideal}.

The results are presented in Table \ref{tab:idl_table}. 
Compared to Sec.~\ref{sec:pref_res}, we
observed more timeouts. This may be due to the size of $\phi$, which would be much larger than the SL formula used to verify membership in a preferred extension, as in Sec.~\ref{sec:pref_res}.
The formula $\phi$ would have also required more
strategy assignments compared to the formula used in the previous
section \cite{CermakP2014MAmc}.

\section{Conclusions} \label{conc}

In this paper, we developed a methodology to verify debates formalised
via abstract argumentation frameworks.  By building on previous works
on debates in abstract argumentation
\cite{dung2007computing,modgil2009proof}, in Sec.~\ref{Sec3} we
introduced a translation from debates
to interpreted systems \cite{Fagin+95b}. Then, in Sec.~\ref{Sec4} we
formalised various winning conditions as formulas in ATL
and SL\cite{MogaveroF.2014RasO}. Finally, in
Sec.~\ref{Sec5} we evaluated the performance of our approach against
state-of-the-art argumentation reasoners.
While the experimental results point to a considerable gap in
performance between the proposed model checking approach and
traditional argumentation reasoners, this is not surprising as model
checker are general purpose tools.  Nonetheless, we deem the present
contribution theoretically relevant under at least two
aspects. Firstly, we provided an automated
translation from debates to
interpreted systems amenable to formal verification. Secondly, we
captured various winning conditions in debates as formulas in
well-known logic-based languages.  We believe that these two
contributions can pave the way for a wider application of verification
techniques in Argumentation Theory, particularly in cases where
argumentation frameworks are more naturally represented as debates.

We plan to develop the results in this paper further. Particularly,
we are interested in formalising winning conditions for other
argumentation semantics, as well as improving the performance of the
verification procedures and tools.


\smallskip
{\bf Acknowledgements.} F.~Belardinelli acknowledges the support of
ANR JCJC Project SVeDaS (ANR-16-CE40-0021).




\bibliographystyle{ACM-Reference-Format}
\bibliography{refs} 
\end{document}